\documentclass[11pt]{article}

\usepackage{fullpage}
\usepackage[utf8]{inputenc}
\usepackage[english]{babel}
\usepackage{amsmath,amssymb,amsthm,stackrel}
\usepackage{todonotes}
\usepackage{authblk}

\newtheorem{lemma}{Lemma}
\newtheorem{theorem}{Theorem}

\newcommand{\GH}{\textsc{Graph Homomorphism}} 
\newcommand{\SI}{\textsc{Subgraph Isomorphism}} 
\newcommand{\LGH}{\textsc{List Graph Homomorphism}} 
\newcommand{\cO}{\mathcal{O}}

\newcommand{\cL}{\mathcal{L}}

\renewcommand{\deg}{\operatorname{deg}}

\newcommand{\gr}[1]{\tilde{#1}} 

\title{Tight Bounds for Subgraph Isomorphism and  Graph Homomorphism\thanks{The research leading to these results has received funding from the Government of the Russian Federation (grant 14.Z50.31.0030). The research of Alexander Kulikov is also supported by  the grant of the President of Russian Federation  (MK-6550.2015.1).
}}
\author[1,4]{Fedor~V.~Fomin}
\author[2,4]{Alexander~Golovnev}
  \author[4]{Alexander~S.~Kulikov}
\author[3,4]{Ivan~Mihajlin}
 
  \affil[1]{University of Bergen, Norway}
    \affil[2]{New York University, USA}
      \affil[3]{University of California---San Diego, USA}
        \affil[4]{St.~Petersburg Department of Steklov Institute of Mathematics of the Russian Academy of Sciences, Russia}
\date{}
\begin{document}
\maketitle
\begin{abstract}
We prove that unless Exponential Time Hypothesis (ETH) fails, deciding if there is a homomorphism from   graph $G$ to  graph $H$ cannot be done in time $|V(H)|^{o(|V(G)|)}$. Combined with the reduction of 
Cygan,   Pachocki,   and Soca{\l}a, our result rules out (subject to ETH) a possibility of $|V(G)|^{o(|V(G)|)}$-time algorithm deciding if graph $H$ is a subgraph of~$G$. 
For both problems our lower bounds asymptotically match the running time of brute-force algorithms trying all possible mappings of one graph into another. Thus, our work  closes  
 the   gap in the known complexity of these fundamental problems.
 \end{abstract}
\section{Introduction}\label{sec:intro}
A {\em homomorphism} $G\to H$ from an undirected graph $G$ to an undirected graph $H$ is a mapping
from the vertex set   $V(G)$ to  $V(H)$ such that the image of every edge of $G$ is an edge of~$H$. Then the \GH{} problem HOM$(G,H)$ is the problem to decide for given graphs $G$ and $H$, whether $G\to H$. \GH {}  is a generic problem and many fundamental combinatorial problems like \textsc{Graph Coloring} and \textsc{Clique} can be seen as its special cases. We refer to books of Hell and Ne\v set\v ril~\cite{HellN04}
and Lov{\'a}sz  \cite{lovasz2012large} for introduction to and applications of graph homomorphisms.

Solving  HOM$(G,H)$ can be done by checking all possible mappings from an $n$-vertex graph $G$ into an $h$-vertex graph~$H$.\footnote{In order to obtain general results, throughout the paper we assume implicitly that $h=h(n)$ is a function of~$n$. We assume that the function $h(n)$ is non-decreasing and time-constructible.}
The running time of this brute-force algorithm is  $\cO(h^n)=2^{\cO(n \log{h})}$. It was shown by Chen et al.~\cite{Chen20061346} that under the ETH assumption, for any constant $\varepsilon>0$, there is no $n^{o(k)}$-time algorithm checking whether a given $n$-vertex graph contains a $k$-clique for any $k=\cO(n^{1-\varepsilon})$. This implies, in particular, that \textsc{Graph Homomorphism} cannot be solved in time $2^{o({n\log{h}})}$ for $h$ significantly larger than~$n$ (again, under the ETH assumption). At the same time this does not exclude the existence of a faster algorithm for some $h \le n$. Moreover, one of the most natural special cases, the $h$-coloring problem (an~$n$ vertex graph can be colored in $h \le n$ colors if and only if there is a homomorphism from the graph to an $h$-clique; for this reason, HOM$(G,H)$ is often called $H$-coloring of~$G$), can be solved in time 
$2^n\operatorname{poly}(n)$ as shown by 
Bj\"{o}rklund et al.~\cite{BHKK2009}. That is why 
the existence of 
an algorithm solving \GH{} asymptotically  faster than the brute-force   was a major open problem in the area of Exact Exponential Algorithms \cite{FHK2007,Rzazewski14,Wahlst10,W2011}. 
In~\cite{FGKM2015}, we have shown that unless ETH fails, there is no algorithm solving \GH{} in time $2^{o\left(\frac{n\log{h}}{\log\log{h}}\right)}$ for every function $h(n)$. In this paper, we close the gap between the existing lower and upper bounds by ruling out a possibility of solving \GH{} in time $2^{o({n\log{h}})}$ for every function $h(n)$.

Our result also implies a tight bound for  the related \SI{} problem. Here, for two given $n$-vertex graphs $G$ and $H$,  the task is to decide if $G$ contains a subgraph isomorphic to~$H$. 
As \GH, \SI{} encompasses many fundamental problems including \textsc{Hamiltonian Cycle}, \textsc{Bandwidth},  \textsc{Triangle Packing},  \textsc{Clique}, and \textsc{Biclique}. 
Again,  a brute-force algorithm solves \SI{}  in time $2^{\cO(n \log{n})}$ and a possibility of time $2^{o(n \log{n})}$
solving \SI{} was another long-standing open question in the area, see for example \cite{AminiFS12,CPS2015,Fomin10}, and  \cite[Chapter 12]{FominKratschbook10}. 
 Recently, 
Cygan,   Pachocki,   and Soca{\l}a~\cite{CPS2015} showed that \GH{} can be solved by solving $2^{\cO(n)}$ instances of \SI. This reduction, together with~\cite{FGKM2015}, implied that \SI{} cannot be solved in time $2^{o\left(\frac{n \log{n}}{\log\log{n}}\right)}$ unless ETH fails~\cite{CPS2015}. Combined with this reduction, our lower bound for homomorphisms rules out algorithms of running time   
 $2^{o(n \log{n})}$ for \SI {}, and closes the gap between upper and lower bounds for this problem as well.

We build the proof of our tight lower bounds for graph homomorphisms  on~\cite{FGKM2015}. As in~\cite{FGKM2015}, we  obtain lower bounds for   \textsc{Graph Homomorphism} by reducing the $3$-coloring problem on graphs of bounded degree to it. The crucial difference with~\cite{FGKM2015}, which allows us to obtain a tight bound, is that we reduce $3$-coloring of graphs of degree $d$ on $n$ vertices to list homomorphism of $\frac{n}{r}$-vertex graph to  $\gamma(d)^r$-vertex graph, where $\gamma(d)$ is a function that depends on $d$ only. Then an  $h^{o(n)}$ upper bound for $n$-vertex to $h$-vertex graph homomorphism would imply a subexponential $\gamma(4)^{o(n)}$ algorithm for 3-coloring on graphs of degree $4$, contradicting ETH. 

\section{Preliminaries}
\subsection{Main Definitions}
Let $G$ be a graph, by $V(G)$ and $E(G)$ we denote the sets of vertices and edges of $G$, respectively. For a vertex $v\in V(G)$, by the neighborhood $N_G(v)$ we mean the set of all vertices of $G$ adjacent to $v$. By the square of $G$ we denote the graph $G^2$, such that $V(G^2)=V(G)$, and $\{u,v\}\in E(G^2)$ if and only if there is a path of length at most two from $u$ to $v$ in $G$.

Let $G$ be an $n$-vertex graph, $1 \le k \le n$ be an integer, and $V(G)=B_1 \sqcup B_2 \sqcup \ldots \sqcup B_k$ be a partition of the set of vertices of~$G$.
Then a {\em grouping} of $G$ with respect to the partition $V(G)=B_1 \sqcup B_2 \sqcup \ldots \sqcup B_k$ is a graph $\gr{G}$
with vertices $B_1, \ldots, B_k$ such that $B_i$ and $B_j$
are adjacent in $\gr{G}$ if and only if there exist $u \in B_i$ and $v \in B_j$ such that $\{u,v\} \in E(G)$. To distinguish vertices of the graphs $G$ and $\gr{G}$, the vertices of $\gr{G}$ will be called {\em buckets}. 

A {\em proper coloring} of a graph $G$ is a function assigning a color to each vertex of $G$ such that adjacent vertices have different colors. 
An {\em equitable coloring} is a proper coloring where the numbers of vertices of any two colors differ by at most one.
An {\em injective coloring} is a proper coloring that assigns different colors to any pair of vertices that have a common neighbor (note that a proper coloring of the square of a graph $G$ is an injective coloring of $G$).

For a positive integer $k$,  we use $[k]$ to denote the set of integers $\{1,\dots, k\}$. All logarithms in this paper are logarithms to the base two.

\subsection{Homomorphism and Subgraph Isomorphism}
 Let $G$ and $H$ be  graphs.  A mapping 
 $\varphi : V(G)\to V(H)$ is a \emph{homomorphism} if for every edge $\{u,v\}\in E(G)$ its image $\{\varphi(u),\varphi(v)\}\in E(H)$.
 If there exists a homomorphism from $G$ to $H$, we write $G\to H$.
  The \GH{} problem  HOM$(G,H)$
  asks whether or not  $G\to H$. We also use the following generalization of graph homomorphism. Assume that for each vertex $v$ of   $G$  there is an assigned  list $\cL(v) \subseteq V (H)$ of vertices. A \emph{list homomorphism} of $G$ to $H$, also known as  a list $H$-colouring of $G$, with respect to the lists $\cL$, is a homomorphism  $\varphi : V(G)\to V(H)$, such that $\varphi(v) \in \cL(v)$ for all $v\in V (G)$. Then the  \LGH{} problem LIST-HOM$(G,H)$ asks whether or not  graph $G$ with lists $\cL$ admits a list homomorphism to $H $ with respect to $\cL$.
  
In the \SI{} problem one is given two graphs $G$ and $H$ and the question is whether $G$ contains a subgraph isomorphic to~$H$.
 
\subsection{Exponential Time Hypothesis} 
Our lower bound is  based on the well-known complexity hypothesis of  Impagliazzo, Paturi, and Zane   \cite{ImpagliazzoPZ01}, see \cite{CFKLMPPS2014,LMS2013} for an overview of  the hypothesis and its implications. 
 
\begin{quote}
\textbf{Exponential Time Hypothesis (ETH)}:  There is a constant $s>0$ such that 3-CNF-SAT with $n$ variables and $m$ clauses cannot be solved in time $2^{sn}(n+m)^{\cO(1)}$.
\end{quote}

 
 Let us remind that in the  \textsc{$3$-Coloring} problem the task is to decide whether a given graph admits a proper coloring in three colors. 
We will need the following folklore lemma. It follows from the fact that  
(unless ETH fails)   \textsc{$3$-Coloring} on  graphs of average degree four cannot be solved in subexponential time, see e.g. Theorem~$3.2$ in~\cite{LMS2013},
and the classical reduction, see e.g. \cite{GareyJ79}, for  \textsc{$3$-Coloring} on degree-four graphs.

\begin{lemma}\label{lemma:3col}
Unless ETH fails, there exists a constant $q>0$  such that there is no algorithm solving \textsc{$3$-Coloring} on $n$-vertex  graphs of maximum degree four in time $\cO\left(2^{q n} \right)$.
\end{lemma}

%

%
%
%


\section{Auxiliary Lemmata}

In this section we provide auxiliary lemmata about colorings which will be used to prove lower bounds for \GH{} and  \SI.
\subsection{Balanced Colorings}
In the following we show how to construct a specific ``balanced" coloring of a graph
in polynomial time. Let $G$ be a graph of constant maximum degree. The coloring of $G$ we want to construct should satisfy three properties. First, it should be a proper coloring of~$G^2$. Then the size of each color class should be bounded as well as the number of edges between vertices from different color classes. More precisiely. 

\begin{lemma}\label{lemma:coloring}
For any constant $d$, there exist constants $\alpha, \beta, \tau>1$ and a polynomial time algorithm
that for a given graph $G$ on $n$ vertices of maximum degree $d$ and an 
integer $\tau \le L \le \frac{n(d^2-1)}{2d^2(d^2+1)}$, finds a coloring $c \colon V(G) \to [L]$ satisfying the following properties:
\begin{enumerate}
\item The coloring $c$ is a proper coloring of~$G^2$.
\item There are only a few vertices of each color: 
for all  $i \in [L]$,
\begin{equation}\label{eq:balver}
|c^{-1}(i)| \le \left\lceil\alpha \cdot \frac{n}{L}\right\rceil \, .
\end{equation}
\item There are only a few  edges of $G$  between each pair of colors:
For all $i \neq j \in [L]$, we have  
\begin{multline*}\label{eq:fg}
k_{i,j} := |\{\{u,v\} \in E (G) \colon c(u)=i, c(v)=j\}| \le 
K_{i,j} := \left\lceil\beta \cdot \frac{\min\{|c^{-1}(i)|, |c^{-1}(j)|\}}{L}\right\rceil \, .
\end{multline*}
\end{enumerate}
\end{lemma}
\begin{proof}
The algorithm starts by constructing greedily an independent set $I$ of $G^2$ of size $\left\lceil\frac{n}{d^2+1}\right\rceil$. 
Since the maximum vertex degree of $G^2$ does not exceed~$d^2$, this is always possible. We construct a partial coloring of $G^2$   by  coloring   the vertices of $I$ in $L$ colors  such that the obtained coloring is an equitable coloring of $G^2[I]$.  Since $I$ is an independent set in $G^2$, such a coloring can be easily constructed in polynomial time.
In the obtained partial equitable coloring, we have that for every $i\in [L]$  
\begin{equation}\label{eq:gh}
|c^{-1}(i)| \ge \left\lfloor\frac{n}{L(d^2+1)}\right\rfloor \ge \frac{n}{2Ld^2} 
\end{equation}
(recall that $L \le \frac{n(d^2-1)}{2d^2(d^2+1)}$).
Let us note that the obtained precoloring of $G^2$ clearly satisfies the   first and the third conditions of the lemma. Since 
  the size of every $c^{-1}(i)$, $i\in [L]$, does not exceed $|c^{-1}(i)| \le \left\lceil  \frac{n}{L}\right\rceil$, the second condition of the lemma also holds for every  $\alpha>1$.

  We extend the precoloring of $G^2$ to the required coloring by the following greedy procedure: We select an uncolored vertex $v$ and   color it by a color from  $[L]$ such that the new partial coloring also satisfies the three conditions of the lemma. In what follows, we prove that such a greedy choice of a color is always possible.

Coloring of a  vertex $v$ with a color $i$  can be {forbidden} only because it breaks one of the three conditions. Let us count, how many colors can be  forbidden for $v$ by each of the three constraints.
\begin{enumerate}
\item Vertex $v$ has at most $d^2$ neighbors in $G^2$, so the first constraint forbids at most $d^2$ colors.
\item The second constraint forbids all the colors that are ``fully packed'' already. The number of such colors is at most $\frac{n}{\left(\frac{\alpha n}{L}\right)}=\frac L\alpha$.
\item To estimate the number of colors forbidden by the third condition,  we go through all the neighbors of~$v$. A neighbor $u \in N_G(v)$ forbids a color $i$
if coloring $v$ by  $i$ exceeds the allowed bound on $k_{i,c(u)}$.
Hence to estimate the number of such forbidden colors $i$ (for every fixed vertex $u$)
we need to estimate how many  values of $k_{i,c(u)}$ can reach the allowed upper bound $K_{i,c(u)}$. We have that 
\begin{align*}
|\{i \colon k_{i,c(u)}  &=  K_{i,c(u)}\}|   \stackrel{\text{by~\eqref{eq:gh}}}{\le}  
\left| \left\{ 
i \colon k_{i,c(u)} \ge \frac{\beta n}{2L^2d^2} 
\right\} \right| = \left| \left\{ 
i \colon k_{i,c(u)} \cdot\frac{2L^2d^2}{\beta n} \ge 1
\right\} \right| & \\  &\leq  \sum_{i \in [L]}k_{i,c(u)}\cdot \frac{2L^2d^2}{\beta n}
. &
\end{align*}
The number of edges between vertices of the same color $c(u)$ and all other vertices of the graph does not exceed the cardinality of the color class $c(u)$ times $d$. Thus we have 

\begin{align*}
 \sum_{i \in [L]}k_{i,c(u)}\cdot \frac{2L^2d^2}{\beta n}
 &
 \le  d|c^{-1}(c(u))|\cdot \frac{2L^2d^2}{\beta n}    \stackrel{\text{by~\eqref{eq:balver}}}{\le} d\left\lceil\frac{\alpha n}{L}\right\rceil\cdot \frac{2L^2d^2}{\beta n} \\ &\le 
d\frac{2\alpha n}{L}\cdot \frac{2L^2d^2}{\beta n}  = 
\frac{4\alpha Ld^3}{\beta} \, .&
\end{align*}

Therefore, 
\begin{align*}
|\{i \colon k_{i,c(u)}  &=  K_{i,c(u)}\}|  \leq \frac{4\alpha Ld^3}{\beta} \, . 
\end{align*}

Since the degree of $v$ in $G$ does not exceed $d$, we have that the  number of colors forbidden by the third constraint is at most $\frac{4\alpha Ld^4}{\beta}$.
\end{enumerate}
Thus, the total number of colors forbidden by all the three constraints for the vertex $v$ is at most
\[d^2 + \frac{L}{\alpha} + \frac{4\alpha L d^4}{\beta} \, .\]
By taking sufficiently large constants 
$\alpha$, $\beta$, and $\tau$, say $\alpha=4$, $\beta=16\alpha^2 d^4$,  and $\tau = 2d^2 +1$, we guarantee that this expression is less than $L-1$
for every $L \ge \tau$. Therefore, there always exists a vacant color for the vertex $v$ which concludes the proof.
\end{proof}

Now with a help of Lemma~\ref{lemma:coloring}, we describe a way to construct a specific grouping of a graph. The properties of such groupings are crucial for the final reduction. 
%

\begin{lemma}\label{lem:group}
For any constant $d$, there exists a constant $\lambda=\lambda(d)$ and a polynomial time algorithm
that for a given graph $G$ on $n$ vertices of maximum degree $d$ and an integer $r\leq \sqrt{\frac{n}{2\lambda}}$, 
 finds a grouping $\gr{G}$ of $G$ and a coloring $\tilde{c} \colon V(\gr{G}) \to [\lambda r]$ such that
\begin{enumerate}
\item The number of buckets of  $\gr{G}$ is 
\[|V(\gr{G})| \le \frac{|V(G)|}{r} \,;\]
\item The coloring $\tilde{c}$ is a proper coloring of $\gr{G}^2$;
\item Each bucket $B \in V(\gr{G})$ is an independent set in~$G$, i.e. for every $u,v\in B$, $\{u,v\} \not\in E(G)$;
\item For every pair of  buckets $B_1,B_2 \in V(\gr{G})$ there is at most one edge between them in~$G$, i.e.
\[|\{\{u,v\} \in E(G) \colon u \in B_1, v \in B_2\}| \le 1 \, .\]
\end{enumerate}
\end{lemma}

\begin{proof} Let $\beta=\beta(d)$ be a constant provided by Lemma~\ref{lemma:coloring} and let $L=\lambda r$ for $\lambda=\lambda(d)=2d\beta$.  Let also $c$ be a coloring of $G$ in $L$ colors provided by Lemma~\ref{lemma:coloring}. We want to construct a grouping $\gr{G}$ of $G$ such that for all buckets $B \in V(\gr{G})$ and all $u \neq v \in B$,
\begin{equation}\label{eq:rz}
c(u)=c(v) \text{ and } c(u') \neq c(v') \text{ for all } u' \in N_G(u), v' \in N_G(v)
\end{equation}
In other words, all  vertices of the same bucket are of  the same color while any two neighbors of such two vertices are of different colors.

For each color $i\in[L]$,  we introduce an auxiliary constraint graph $F_i$. The vertex set of $F_i$ is  $V(F_i)=c^{-1}(i)$ and its edge set is 
\begin{equation}
E(F_i) = \{\{u,v\} \colon \exists u'\in N_G(u),v' \in N_G(v), c(u')=c(v') \}. 
\end{equation}
In our construction, each bucket of $\tilde{G}$ will be an independent set in some $F_i$. Note that this will immediately imply~(\ref{eq:rz}). The degree of any vertex $v\in V(F_i)$ is at most
\begin{equation*}
\deg_{F_i}(v) \leq \sum_{v' \in N_G(v)} (K_{c(v),c(v')}-1) \leq d \left(\left\lceil\frac{\beta |c^{-1}(v)|}{L}\right\rceil-1\right) \leq \frac{d \beta |V(F_i)|}{L} = \frac{|V(F_i)|}{2 r}\,.
\end{equation*}
This means that the greedy algorithm finds a proper coloring of each $F_i$ in at most $\frac{|V(F_i)|}{2r}+1$ colors, which splits each $F_i$ in at most $\frac{|V(F_i)|}{2r}+1$ independent sets. We create a separate bucket of $\gr{G}$ from each independent set of each $F_i$. Now we show that the four conditions from the lemma statement hold.
\begin{enumerate}
\item  For the first property, the number of independent sets in each $F_i$ is at most $\frac{|V(F_i)|}{2r}+1$. Thus the number of buckets in $\gr{G}$ is 
\[ |V(\gr{G})|\leq \sum_{i\in[L]} \left(\frac{|V(F_i)|}{2r}+1\right) = \sum_{i\in[L]} \left(\frac{|c^{-1}(i)|}{2r}+1\right) = \frac{n}{2r}+L\le\frac{n}{r} \,,\]
since $L=\lambda r$ and $2\lambda r^2 \le n$.
\item For the second property, by Lemma~\ref{lemma:coloring}, the coloring $c$ is  proper in~$G^2$.
We can convert $c$ to a coloring $\tilde{c} \colon V(\gr{G}) \to [\lambda r]$ by assigning each bucket the color of its vertices (all of them have the same color). The resulting coloring $\tilde{c}$ is a proper coloring of $\gr{G}^2$ by \eqref{eq:rz}.
\item All buckets of $\gr{G}$ are monochromatic with respect to $c$, thus, each bucket $B \in V(\gr{G})$ is an independent set in~$G$ and the third property holds.
\item Finally, by (\ref{eq:rz}), there is at most one edge in $G$ between vertices corresponding to any pair of buckets in~$\gr{G}$.
\end{enumerate}
Thus, the constructed grouping and its coloring satisfy all conditions of the lemma.
\end{proof}

\begin{lemma}\label{lem:coltolhom}
There exists a polynomial time algorithm that takes an input a graph $G$ on $n$ vertices of maximum degree $d$ 
that needs to be $3$-colored and an integer $r=o(\sqrt{n})$ and finds an equisatisfiable instance $(G',H')$ of \LGH, where $|V(G')|\le\frac{n}{r},|V(H')|\le\gamma(d)^r$.
\end{lemma}
\begin{proof}
{\em Constructing the graph $G'$.}
Let $G'$ be the grouping of $G$ and $c\colon V(G') \to [L]$ be the coloring provided by Lemma~\ref{lem:group} where~$L=\lambda(d) r$. To distinguish colorings of $G$ and $G'$, 
we call $c(B)$, for a bucket $B \in V(G')$, a~\emph{label} of~$B$.
Consider a bucket $B \in V(G')$, i.e., a subset of vertices of~$G$, and a label $i \in [L]$. From item 2 of Lemma~\ref{lem:group} we know that $c$ is a proper coloring of $(G')^2$. This, in particular, means that there is at most one $B' \in N_{G'}(B)$ such that $c(B')=i$. Moreover, if such $B'$ exists then, by item 4 of Lemma~\ref{lem:group}, there exists a unique $u \in B$ and unique $u' \in B'$ such that $\{u,u'\} \in E(G)$. This allows us to define the following mapping $\phi_B \colon [L] \to B \cup \{0\}$: $\phi_B(i)=u$ if such $B'$ exists and $\phi_B(i)=0$ if $B$ has no neighbor $B'$ of label~$i$. 

{\em Constructing the graph $H'$.}
We now define a redundant encoding of a $3$-coloring of a bucket $B \in V(G')$. Namely, let $\mu_B \colon (f \colon B \to \{1,2,3\}) \to \{0,1,2,3\}^L$.
That is, for a $3$-coloring $f \colon B \to \{1,2,3\}$ of $B$, $\mu_B$ is a vector $v$ of length~$L$. For $i \in [L]$, by $v[i]$ we denote the $i$-th component of~$v$. The value of $v[i]$ is defined as follows:
if $\phi_B(i)=0$ then $v[i]=0$, otherwise $v[i]=f(\phi_B(i))$.
In~other words, for a given bucket $B$ and a $3$-coloring $f$ of its vertices, for each possible label~$i \in [L]$, $\mu_B$
is the color of a vertex $u \in B$ that has a neighbor in a bucket with label $i$, and $0$ if there is no such vertex~$u$. 
  

We are now ready to construct the graph $H'$.
The set of vertices of $H'$ is defined as follows:
\[ V(H') =\{(R,l) \colon  R \in \{0,1,2,3\}^L \text{ and } l \in [L] \}\,,\] i.e., a~vertex of $H'$ is an encoding of a $3$-coloring of a bucket and a label of a bucket. 
A~bucket $B \in V(G')$ is allowed to be mapped to $(R,l) \in V(H')$ if and only if $l=c(B)$ and there is a $3$-coloring $f$ of $B$ such that $\mu_B(f)=R$.
Informally, two vertices in $V(H')$ are joined by an edge 
if they define two consistent $3$-colorings. Formally, 
$\{(R_1,l_1), (R_2, l_2)\} \in E(H')$ if and only if
$R_1[l_2] \neq R_2[l_1]$. Note that $|V(G')| \le n/r$ by Lemma~\ref{lem:group} and
$|V(H')| \le 4^L \cdot L \le 5^L=5^{\lambda(d)r}=\gamma(d)^r$
for $\gamma(d)=5^{\lambda(d)}$. 

{\em Running time of the reduction.} The reduction clearly takes polynomial time.

{\em Correctness of the reduction.} It remains to show that $G$ is $3$-colorable if and only if $(G',H')$ is a yes-instance of \LGH.

Assume that $G$ is $3$-colorable and take a proper $3$-coloring $g$ of~$G$. It defines a homomorphism from $G'$ to $H'$ in a natural way: $B \in V(G')$ is mapped to $(\mu_B(g), l(B))$.
Each list constraint is satisfied   by definition. To show that each edge is mapped to an edge, consider an edge $\{B,B'\} \in E(G')$. Then, by item 4 of Lemma~\ref{lem:group} there is a unique edge $\{u,u'\} \in E(G)$ such that $u \in B, u' \in B'$.
Note that $B$ and $B'$ are mapped to vertices $(R,l)$ and $(R',l')$ such that $R[l']=g(u)$ and $R'[l]=g(u')$. Since $g$
is a proper $3$-coloring of $G$, $g(u) \neq g(u')$. This, in turn, means that $\{(R,l), (R',l')\} \in E(H')$ and hence the edge $\{B,B'\}$ is mapped to this edge in~$H'$.

For the reverse direction, consider a homomorphism $h \colon G' \to H'$. For each bucket $B \in V(G')$, $h(B)$ defines a proper $3$-coloring of~$B$. Together, they define a $3$-coloring $g$ of $G$ and we need to show that $g$ is proper. Assume, to the contrary, that there is an edge $\{u,u'\} \in E(G)$ such that $g(u)=g(u')$. By item 3 of Lemma~\ref{lem:group}, $u$ and $u'$
belong to different buckets $B, B' \in V(G')$. By the definition of grouping, $\{B,B'\} \in E(G')$. Since $h$ is a homomorphism,
$\{(R,l), (R',l')\} := \{h(B),h(B')\} \in E(H')$. At the same time, $R[l']=g(u)=g(u')=R'[l]$ which contradicts to the fact that 
$\{(R,l), (R',l')\}$ is an edge in~$H'$.
\end{proof}

\section{Main Theorems}\label{sec:maintheorems}
\subsection{Graph Homomorphism}
For the proof of the first main theorem of this paper, we need 
the following lemma which is proved in ~\cite[Lemma~5]{FGKM2015}. 

\begin{lemma}
\label{lemma:lhomtohom}
There is a polynomial-time algorithm that from an instance $(G,H)$ of \LGH{} where $|V(G)|=n$, $|V(H)|=h\ge3$, constructs an   instance $(G',H')$ of \GH, where $|V(G')| \le n+s$ and  $|V(H')| \le s$ for $s 
<25h^2$, such that
there is a list homomorphism from $G$ to $H$ if and only if there is a homomorphism from $G'$ to $H'$. 

\end{lemma}

We are ready to prove our main theorem about graph homomorphisms.

\begin{theorem}\label{thm:main}
Let $G$ be an $n$-vertex graph $G$ and $H$ be an $h(n)$-vertex graph.
Unless ETH fails, for any constant $D\ge1$ there exists a constant $c=c(D)>0$ such that
for any function
$3\le h(n)\le n^D$, there is no time $\cO\left(h^{cn}\right)$ algorithm deciding whether there is a homomorphism from $G$ to $H$. 
\end{theorem}
\begin{proof}
The outline of the proof of the theorem is as follows. Assuming that there is a ``fast" algorithm for \GH, we show that there is  also a  ``fast"  algorithm solving \LGH, which, in turn, implies  ``fast" algorithm for \textsc{3-Coloring} on degree $4$ graphs, contradicting ETH. In what follows, we specify what we mean by ``fast".

Let $h_0=25^2$. If $h(n)<h_0$ for all values of $n$, then an algorithm with running time $\cO\left(h^{cn}\right)$ would solve \textsc{3-Coloring} in time $\cO\left(h_0^{cn}\right)=\cO\left(2^{cn\log{h_0}}\right)$ (recall that $h(n) \ge 3$). Therefore, by choosing a small enough constant $c$ such that $c\log{h_0} < q$, we arrive to a contradiction with Lemma~\ref{lemma:3col}.

From now on we assume that $h(n)\ge h_0$ for large enough values of $n$.
Let $c=\frac{q}{8D\log{\gamma}}$, where $q$ is the constant from Lemma~\ref{lemma:3col}, and $\gamma:=\gamma(4)$ is the constant from Lemma~\ref{lem:coltolhom}. For 
  the sake of contradiction, let us assume that there exists an algorithm $\cal A$ deciding whether  $G\to H$ in time $\cO(h^{cn})=\cO(2^{cn\log{h}})$, where $|V(G)|=n, |V(H)|=h:=h(n)$. Now we show how to solve $3$-coloring on $n'$-vertex graphs of degree $4$ in time $2^{q n'}$, which would contradict  Lemma~\ref{lemma:3col}.
 
Let $r=\frac{\log{h}}{4D\log{\gamma}}$ and $n'=\frac{nr}{2}$. Let $G'$ be an $n'$-vertex graph of maximum degree four that needs to be $3$-colored. Using Lemma~\ref{lem:coltolhom} we construct an instance $(G_1,H_1)$ of \LGH{} that is satisfiable if and only if the initial graph $G'$ is $3$-colorable, and $|V(G_1)|\le\frac{n'}{r}, |V(H_1)|\le\gamma^r$. By Lemma~\ref{lemma:lhomtohom}, this instance is equisatisfiable to an instance $(G,H)$ of \GH{} where $|V(H)|<25\gamma^{2r}\leq 25h^{\frac{1}{2D}}\le h$ (since $D\ge1$ and $h(n)\ge h_0$), and 
\[|V(G)|< \frac{n'}{r}+25\gamma^{2r} \le \frac{n}{2}+25h^{\frac{1}{2D}}\le \frac{n}{2}+25\sqrt{n}\le n  \]
(for sufficiently large values of~$n$).

Now, in order to solve $3$-coloring for $G'$, we construct an instance $(G,H)$ with $|V(G)|\le n$ and $|V(H)|\le h$ of 
\GH \,  and 
invoke the algorithm $\cal A$ on this instance.
 The running time of $\cal A$ is
\[ \cO( 2^{cn\log{h}}) =\cO(2^{\frac{2cn'}{r}\log{h}} )= \cO(2^{2cn'\log{h}\cdot\frac{4D\log{\gamma}}{\log{h}}} )= \cO(2^{8cDn'\log{\gamma}} )= \cO(2^{q n'}) \, \]
and hence we can find a 3-coloring of $G'$ in time $\cO(2^{q n'})$, which contradicts ETH. 
\end{proof}

\subsection{Subgraph Isomorphism}\label{sec:GI}

To prove the theorem about  subgraph isomorphisms, we need
the following theorem of   Cygan,   Pachocki,   and Soca{\l}a  (Theorem~1.3 in~\cite{CPS2015}).
\begin{theorem}\label{thm:homtoiso}
Given an instance $(G, H)$ of \GH{} one can in $\operatorname{poly}(n)2^n$ time create $2^n$ instances of \SI{} with $n$ vertices, where $n = |V(G)|+|V(H)|$, such that $(G, H)$ is a yes-instance if and only if at least one of the created instances of \SI{}  is a yes-instance.
\end{theorem}

 Combining Theorem~\ref{thm:main} with Theorem~\ref{thm:homtoiso}, we immediately obtain the following theorem.
\begin{theorem}\label{thmSI}
Unless ETH fails, there exists a constant $c>0$ such that there is no algorithm deciding whether a given $n$-vertex graph $G$ contains a subgraph isomorphic to a given $n$-vertex graph $H$ in time 
$\cO\left(n^{cn}\right)$.
\end{theorem}

\bibliographystyle{plain}
\bibliography{hom}

\end{document}